\documentclass[a4paper,reqno,twoside]{amsart}

\usepackage[%
pdfauthor={Silvio Barandun,%
            },%
pdftitle={Anderson transition in disordered Hatano-Nelson systems},
]{hyperref}

\usepackage[left=3.5cm,right=3.5cm,top=3.5cm,bottom=3.5cm,footskip=1.5cm,headsep=1cm,headheight=12pt,bindingoffset=0.cm,]{geometry}

\DeclareMathOperator{\R}{\mathbb{R}}
\DeclareMathOperator{\C}{\mathbb{C}}
\DeclareMathOperator{\sgn}{sign}

\renewcommand{\i}{\mathbf{i}}
\renewcommand{\tilde}{\widetilde}
\renewcommand{\hat}{\widehat}
\renewcommand{\bar}[1]{\overline{#1}}

\renewcommand{\qedsymbol}{$\blacksquare$}

\renewcommand{\epsilon}{\varepsilon}
\DeclareMathOperator{\dd}{d\!}

\renewcommand{\i}{\mathbf{i}}
\renewcommand{\tilde}{\widetilde}
\renewcommand{\hat}{\widehat}

\usepackage{tcolorbox}
\usetikzlibrary{tikzmark,calc}

\usepackage{tikz}

\usetikzlibrary{matrix} 
\usetikzlibrary{arrows,arrows.meta,bending} 

\usepackage{nicematrix}
\usepackage{calc} 

\DeclareMathOperator{\arcosh}{arcosh}

\newcommand{\hhn}{H_{\text{HN}}}
\DeclareMathOperator*{\var}{Var}
\DeclareMathOperator*{\E}{\mathbb{E}}

\usepackage[utf8]{inputenc}
\usepackage[english]{babel}
\usepackage[T1]{fontenc} 
\usepackage{lmodern} 
\rmfamily 
\DeclareFontShape{T1}{lmr}{b}{sc}{<->ssub*cmr/bx/sc}{}
\DeclareFontShape{T1}{lmr}{bx}{sc}{<->ssub*cmr/bx/sc}{}

\numberwithin{equation}{section}
\usepackage{amsmath}
\usepackage{amssymb}
\usepackage{amsthm}
\usepackage{amsfonts}
\usepackage{mathtools}

\usepackage{mathdots}
\usepackage{caption}
\usepackage{subcaption}

\usepackage{dsfont} 
\usepackage[format=hang]{caption}
\usepackage[normalem]{ulem}

\usepackage{standalone}
\usepackage{graphicx}
\usepackage{imakeidx}
\makeindex
\usepackage{setspace}
\usepackage{appendix}
\usepackage{pdfpages}
\usepackage{upgreek}
\usepackage{tabulary}
\usepackage{booktabs}
\usepackage{extarrows}
\usepackage{tabularx}
\usepackage{float}
\restylefloat{table}
\usepackage{enumitem}
\usepackage{csquotes}
\usepackage{bm}
\usepackage{orcidlink}
\usepackage{lipsum}
\usepackage{nicematrix}

\usepackage{tikz}
\usepackage{tikz-layers}
\usepackage{tikz-cd}
\usepackage[arrow, matrix, curve]{xy}
\usetikzlibrary{matrix,positioning,decorations.pathreplacing,calc}
\usetikzlibrary{
    decorations.text,%
    decorations.markings,%
    shadows}
\usetikzlibrary{calc,intersections}
\usepackage{adjustbox}
\usepackage{graphicx}

\usepackage[
    style=numeric,
    sorting=nty,
    maxnames=99,
    maxalphanames=5,
    natbib=true,
    backend=bibtex,
    sortcites]{biblatex}

\DeclareNameAlias{default}{family-given}

\AtEveryBibitem{
    \clearfield{url}
    \clearfield{issn}
    \clearfield{isbn}
    \clearfield{urldate}

    \ifentrytype{book}{
        \clearfield{pages}}{%
    }
}

\usepackage{xargs}                      
\usepackage[prependcaption,textsize=tiny,textwidth=3cm]{todonotes}
\newcommandx{\unsure}[2][1=]{\todo[linecolor=red,backgroundcolor=red!25,bordercolor=red,#1]{#2}}
\newcommandx{\change}[2][1=]{\todo[linecolor=blue,backgroundcolor=blue!25,bordercolor=blue,#1]{#2}}
\newcommandx{\info}[2][1=]{\todo[linecolor=OliveGreen,backgroundcolor=OliveGreen!25,bordercolor=OliveGreen,#1]{#2}}
\newcommandx{\improvement}[2][1=]{\todo[linecolor=black,backgroundcolor=black!25,bordercolor=black,#1]{#2}}
\newcommandx{\thiswillnotshow}[2][1=]{\todo[disable,#1]{#2}}

\usepackage{aliascnt}
\usepackage[noabbrev,capitalize]{cleveref}
\crefname{proposition}{Proposition}{Propositions}
\crefname{equation}{}{}

\newtheorem{theorem}{Theorem}[section]
\newaliascnt{lemma}{theorem}

\aliascntresetthe{lemma}
\newaliascnt{proposition}{theorem}
\newtheorem{proposition}[proposition]{Proposition}
\aliascntresetthe{proposition}
\newaliascnt{corollary}{theorem}

\aliascntresetthe{corollary}
\newaliascnt{conjecture}{theorem}

\aliascntresetthe{conjecture}

\theoremstyle{definition}
\newaliascnt{definition}{theorem}
\newtheorem{definition}[definition]{Definition}
\aliascntresetthe{definition}
\newaliascnt{example}{theorem}

\aliascntresetthe{example}
\newaliascnt{assumption}{theorem}

\aliascntresetthe{assumption}
\newaliascnt{remark}{theorem}

\aliascntresetthe{remark}

\crefname{assumption}{Assumption}{Assumptions}
\crefname{definition}{Definition}{Definitions}
\Crefname{definition}{Definition}{Definitions}
\crefname{corollary}{Corollary}{Corollaries}
\crefname{enumi}{item}{items}

\creflabelformat{subfigure}{#2\textsc{#1}#3}

\usepackage[final]{microtype}

\makeatletter
\newsavebox\myboxA
\newsavebox\myboxB
\newlength\mylenA

\newcommand*\xoverline[2][0.75]{%
  \sbox{\myboxA}{$\m@th#2$}%
  \setbox\myboxB\null
  \ht\myboxB=\ht\myboxA%
  \dp\myboxB=\dp\myboxA%
  \wd\myboxB=#1\wd\myboxA
  \sbox\myboxB{$\m@th\overline{\copy\myboxB}$}
  \setlength\mylenA{\the\wd\myboxA}
  \addtolength\mylenA{-\the\wd\myboxB}%
  \ifdim\wd\myboxB<\wd\myboxA%
    \rlap{\hskip 0.5\mylenA\usebox\myboxB}{\usebox\myboxA}%
  \else
    \hskip -0.5\mylenA\rlap{\usebox\myboxA}{\hskip 0.5\mylenA\usebox\myboxB}%
  \fi}
\makeatother

\usepackage{fancyhdr}

\fancypagestyle{plain}{%
    \fancyhead[C]{}
    \fancyfoot[C]{}
    
    }
\fancypagestyle{nosection}{%
    \fancyhead[CE]{}
    \fancyhead[CO]{}
    \fancyfoot[CE,CO]{\thepage}
}

\newcommandx{\silvio}[2][1=]{\todo[linecolor=blue,backgroundcolor=blue!25,bordercolor=blue,#1]{Silvio: #2}}

\title{Anderson transition in disordered Hatano-Nelson systems}

\begin{document}

\author[S. Barandun]{Silvio Barandun\,\orcidlink{0000-0003-1499-4352}}
\thanks{Massachusetts Institute of Technology, Department of Mathematics, United States of America, \href{http://orcid.org/0000-0003-1499-4352}{orcid.org/0000-0003-1499-4352}}
  \address{\parbox{\linewidth}{Silvio Barandun\\
 Massachusetts Institute of Technology, Department of Mathematics, Simons Building (Building 2), 77 Massachusetts Avenue, Cambridge, MA 02139-4307, United States of America, \href{http://orcid.org/0000-0003-1499-4352}{orcid.org/0000-0003-1499-4352}}.}
 \email{barandun@mit.edu}

\begin{abstract}
    We illuminate the fundamental mechanism responsible for the transition between the non-Hermitian skin effect and defect-induced Anderson localization in the bulk via the study of Lyapunov exponents. We obtain a proof that the change of the topological invariant associated with an eigenvalue coincides with the eigenvector crossover from non-Hermitian skin effect to Anderson localization, establishing a universal criterion for localization behavior.
\end{abstract}

\maketitle






\section{Introduction}\label{sec:intro}
In this article we study the transition between two localization regimes: the non-Hermitian skin effect (NHSE) and Anderson localization. Anderson localization describes the phenomenon of localization due to disorder. Firstly introduced by the seminal paper \cite{anderson1958Absence}, the study of disorder induced localization has rapidly developed becoming one of the hottest topics in mathematical research. On the other hand, the NHSE describes the localization phenomena whereas all eigenmodes of a system are ``condensed'', that is localized at one edge of the system. Firstly introduced by \citeauthor{hatano.nelson1996Localization} \cite{hatano.nelson1996Localization,hatano.nelson1998NonHermitian} in \citeyear{hatano.nelson1996Localization} as a non-Hermitian equivalent of Anderson localization, this kind of non-Hermitian localization has mostly moved away from the study of random or disordered system due to the rich mathematical theory for non-Hermitian crystalline or periodic system \cite{Yao_2018,borgnia.kruchkov.ea2020Nonhermitian,li.lee.ea2020Critical,lin.tai.ea2023Topological,maddi.auregan.ea2024Exact,zhang.zhang.ea2022Review,ammari.barandun.ea2024Mathematical}. Specifically, the link to topological invariants that uniquely describe the localization behavior has overwhelmingly dominated the subject in the past decade.

We present here for the first time full characterization of the transition from a NHSE dominated localization in a crystalline material to an Anderson disordered system via \emph{a priori}\footnote{That is, of the unperturbed structure.} information.
\begin{figure}[h]
    \includegraphics[width=0.5\textwidth]{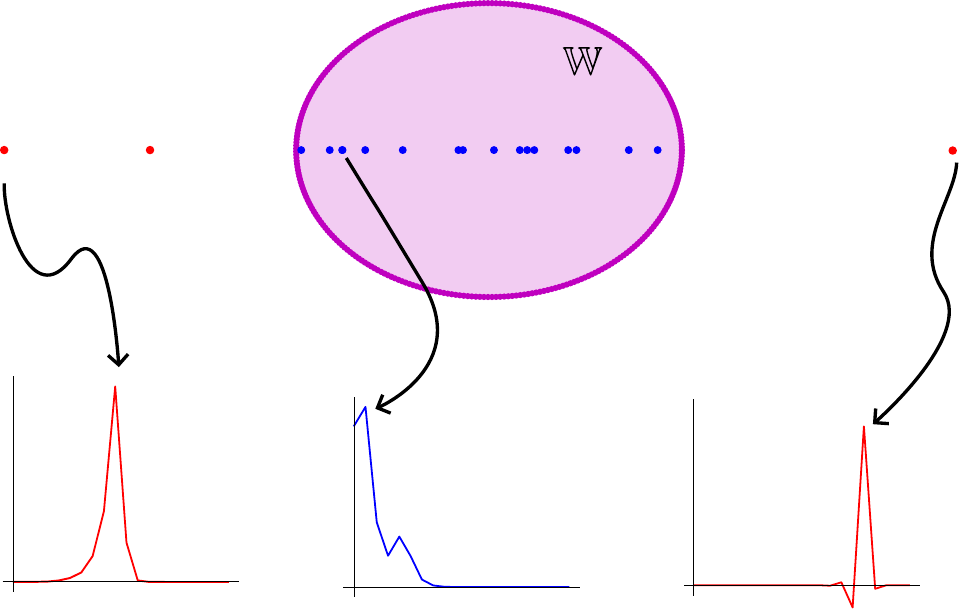}
    \caption{The transition of an eigenvalue from inside to outside the topological winding region $\mathbb{W}$ coincides with the eigenvector crossover from non-Hermitian skin effect to Anderson localization, establishing a universal criterion for localization behavior.}
    \label{fig:intro}
\end{figure}
As mentioned in the previous paragraph, the Hatano-Nelson (HN) system was initially conceived in a disordered setting and research has been conducted in this view \cite{goldsheid.khoruzhenko1998Distribution,trefethen.embree2005Spectra,Hatano_2021,Fortin_2025,Wang_2025}. However, we are, to the best of our knowledge, the first to link via the \emph{a priori} topological invariant the transition from one regime to the other in the fully disordered case. We want to mention the work \cite{trefethen.embree2005Spectra, trefethen.contedini.ea2001Spectra} which introduces useful \emph{a posteriori}\footnote{That is, of an already perturbed structure.} classification of the eigenvalues into four groups. This work is probably the closest in scope to ours. Previous work has provided a first insight into this precise topic \cite{davies.barandun.ea2025Twoscale}, but limited to a single defect. The main idea resides in \cref{fig:intro} and was conceived as a conjecture during the preparation of the author's earlier paper \cite{ammari.barandun.ea2024Stability}.

In broad terms we conjectured that the transition of an eigenvector from NHSE to Anderson localization happens at the same time as the transition of the associated eigenvalue from inside to outside of the topological region of the unperturbed structure ($\mathbb{W}$ in \cref{fig:intro}). We are able to provide quantitative results proving this statement. It is interesting to notice that there is a fundamental difference between Hermitian and non-Hermitian system: while in Hermitian system any (also asymptotically small) non-zero noise induces Anderson localization, for non-Hermitian system there is a non-zero minimum noise amplitude needed to induce Anderson localization, due to the stability of the NHSE \cite{ammari.barandun.ea2024Stability}. Our result proves the existence and quantifies this minimal value.

The paper is structured as follows:
\begin{description}
    \item[\cref{sec:not}] introduces the key objects of consideration of the paper, that is Lyapunov exponents and the Hatano-Nelson system;
    \item[\cref{sec:Lloyd model}] analyzes the Lloyd model and provides with \cref{thm: main Lloyd} our strongest result;
    \item[\cref{sec:gen asymptotic}] looks at the more general case in the weak-disorder regime. 
\end{description}

\section{Notation}\label{sec:not}
The main object of study of this paper is the HN Hamiltonian
\begin{align}
    \hhn^\gamma = \sum_i e^{-\gamma}\vert i+1\rangle\langle i\vert + e^{\gamma}\vert i\rangle\langle i+1\vert + V_i \vert i\rangle\langle i\vert \label{eq: HN Hamiltonian}
 \end{align}
where $V_i$ are \emph{i.i.d.} random variables with a given distribution.

The standard tool used to understand localization in tight-binding Hamiltonians as \eqref{eq: HN Hamiltonian} (especially when $\gamma=0$) is the \emph{Lyapunov exponent}. The Lyapunov exponent is typically introduced via transfer matrices.
\begin{definition}[Lyapunov exponent]\label{def: lyapunov}
    Considering the tridiagonal matrix
    \begin{align*}
        J = \begin{pmatrix}
        \ddots & \ddots\\
        \ddots &\ddots &\ddots \\
        & a_{-1} & b_{-1} & c_{-1}\\
        && a_0 & b_0 & c_0\\
        &&& a_{1} & b_{1} & c_{1}\\
        &&&&\ddots &\ddots &\ddots \\
        &&&&&\ddots &\ddots &
    \end{pmatrix},
    \end{align*}
    the $j$-th associated transfer matrix is given by
    \begin{align*}
        t_j(z) =\begin{pmatrix}
            \frac{z - b_j}{c_j}& -\frac{a_j}{c_j}\\
            1 & 0
        \end{pmatrix}.
    \end{align*}
    The Lyapunov exponent at frequency $z\in \C$ is then given by the limit
    \begin{align}
        L(z) = \lim_{n\to\infty}\frac{1}{2n}\log\left(\left\Vert \prod_{\vert i\vert < n } t_j(z)\right\Vert\right).
        \label{eq: def laypunov}
    \end{align}
\end{definition}
A series of results ensures the existence of the Lyapunov exponent in a variety of cases, see \cite[Chapter V]{pastur.figotin1992Spectra} for a summary. In particular, these studies have been overwhelmingly focused to the Hermitian case, that is $a_i=\bar{c_i}$. This is not the case of $\hhn$ of \eqref{eq: HN Hamiltonian}. However, \cref{def: lyapunov} works perfectly for any tridiagonal matrix and we will use this definition also for $\hhn$.

In the case of random matrices, it is typical to take expected value
\begin{align*}
    L(z) = \lim_{n\to\infty}\frac{1}{2n}\mathbb{E}\left[\log\left(\left\Vert \prod_{\vert i\vert < n } t_j(z)\right\Vert\right)\right],
\end{align*}
in the same notation as \cref{def: lyapunov}

One might notice that using the well-known similarity of a tridiagonal matrix $J$ as in \cref{def: lyapunov} to a symmetric matrix (under the assumption that $a_ic_{i-1}>0$, which we implicitly make) the Lyapunov exponent behaves as follows:
\begin{align}
    L_\gamma(z) =  L_0(z) - \gamma,
    \label{eq: layp herm to not herm}
\end{align}
where $ L_\gamma(z)$ is the Lyapunov exponent associated to $\hhn^\gamma$. This can be seen directly by looking at propagation matrices as in \cite{ammari.barandun.ea2025Competing}

The sign of the Lyapunov exponent determines the localization behavior of the associated eigenvectors. For the non-Hermitian $\gamma\neq 0$ case, the Lyapunov exponent can take negative values due to the $-\gamma$ shift in \eqref{eq: layp herm to not herm}. When $L_\gamma(z) < 0$, the transfer matrix products in \cref{def: lyapunov} decay exponentially, implying that solutions to the eigenvalue equation grow in one direction and decay in the opposite direction, which is the precise signature of the NHSE. The negative Lyapunov exponent quantifies the rate at which this boundary localization occurs.

Conversely, when $L_\gamma(z) > 0$, the transfer matrices grow exponentially in both directions, forcing eigenvectors to decay away from any given point to maintain normalizability. This is precisely the mechanism of Anderson localization, where disorder-induced scattering causes eigenstates to localize around random positions in the bulk of the system, with the Lyapunov exponent measuring the inverse localization length.

In the case of constant deterministic potentials, that is $V_i=V\in \R$, the infinite matrix associated to \eqref{eq: HN Hamiltonian} is Toeplitz making it a \emph{Laurent operator}. There is a rich theory associated to Toeplitz matrices, operators and Laurent operators \cite{trefethen.embree2005Spectra} which lies at the core of the topological invariants of non-Hermitian operators as $\hhn$. In particular one might associate to $\hhn^\gamma$ a \emph{symbol function}
\begin{align*}
    f_\gamma : S^1 &\to \C\\
        z &\mapsto e^{\gamma}z + V + e^{-\gamma}z^{-1}
\end{align*}
for $S^1\subset \C$ the unit circle. A well-known result associates the winding region of $f_\gamma$
\begin{align}
    \mathbb{W}\coloneqq \left\{z\in\C: \frac{1}{2\pi \i} \int_{f_\gamma} \frac{1}{\zeta-z} \dd \zeta \neq 0\right\},
    \label{eq: winding region}
\end{align}
to an exponential decay of the eigenvectors of $\hhn$ \cite[Theorem 7.2]{trefethen.embree2005Spectra}. Specifically it says that for a finite truncation of $\hhn$ of size $N$ any value $\lambda \in \mathbb{W}$ is an exponentially (in the size of the truncation) good pseudo-eigenvalue and its associated eigenvector is exponentially decaying, in formulas: there exits a $\epsilon >0$ and $N$ large enough
\begin{align*}
    \Vert (\lambda - \hhn\vert_{N\times N})^{-1}\Vert \geq \epsilon^{-N}
\end{align*}
where $\hhn\vert_{N\times N}\in \R^{N\times N}$ is the central truncation of $\hhn$. Furthermore there exits a $v\in\R^N$ of unit length $\Vert v\Vert =1$ so that
\begin{align*}
    \Vert (\hhn\vert_{N\times N} - \lambda) v\Vert \leq \epsilon^N
\end{align*}
and 
\begin{align*}
    \frac{\vert v_j \vert}{\max_j \vert v_j \vert }\leq \begin{dcases}
        \epsilon^j&\text{if}\quad \frac{1}{2\pi \i} \int_{f_\gamma} \frac{1}{\zeta-\lambda} \dd \zeta <0\\
        \epsilon^{N-j}&\text{if}\quad \frac{1}{2\pi \i} \int_{f_\gamma} \frac{1}{\zeta-\lambda} \dd \zeta >0
    \end{dcases}.
\end{align*}

\section{The non-Hermitian Lloyd Model}\label{sec:Lloyd model}
We consider now a HN Hamiltonian with \emph{i.i.d.} Cauchy distributed potentials with zero location (WLOG), that is
\begin{align}
    \mathbb{P}[V_0 \in \Omega] = \frac{1}{\pi s}\int_{\Omega}\frac{1}{1+\frac{x^2}{s^2}}\dd x.\label{eq: cauchy pdf}
\end{align}
\cref{fig: cdfs} shows the cumulative distribution function of \cref{eq: cauchy pdf} for some values of $s$.

\begin{figure}
    \includegraphics[width=0.5\textwidth]{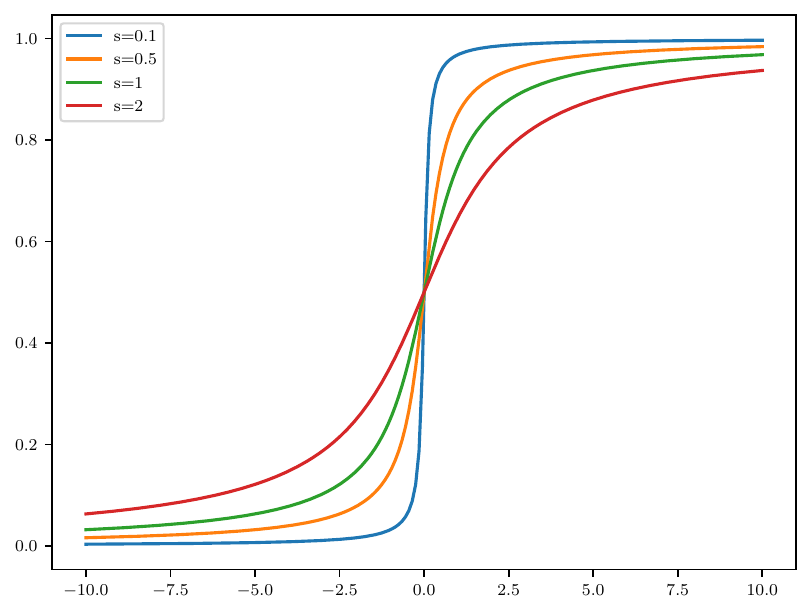}
    \caption{Cumulative distribution function for the Cauchy distribution for different scale parameters $s$.}
    \label{fig: cdfs}
\end{figure}

This model has been studied in the Hermitian case ($\gamma=0$) already 25 years before \cite{hatano.nelson1996Localization} by \citeauthor{lloyd1969Exactly} in \cite{lloyd1969Exactly}. We condense results from \cite{lloyd1969Exactly,thouless1972Relation}.

Let $\rho(z)$ be the density of states of the Hamiltonian $\hhn^0$ and define the complexified Lyapunov exponent
\begin{align}
    w(z) = \int_{-\infty}^\infty \rho(z)\log( z-x)\dd x,
    \label{eq: complexified layp}
\end{align}
which by the Thouless formula \cite{goldsheid.khoruzhenko2005Thouless} satisfies $\Re(w(z)) = L_0(z)$. The complexified Lyapunov exponent computed for the Lloyd model satisfies
\begin{align*}
    \cosh(w(z)) = \frac{1}{4}\left(\sqrt{(z+2)^2+s^2}+\sqrt{(z-2)^2+s^2}\right)
\end{align*}
which is well-defined because the Lyapunov exponent is non-negative for Hermitian matrices. Some careful algebraic manipulation can transform this into
\begin{align}
    L_0(z) = \Re\left(\arcosh\left(\frac{z+\i s}{2}\right)\right).
    \label{eq: L0 Lloyd}
\end{align}

It remains to analyze the relation between Lyapunov exponent and $\mathbb{W}$. We recall from \eqref{eq: winding region} that 
\begin{align*}
    \mathbb{W} = \left\{z\in\C: \frac{1}{2\pi \i} \int_{f_\gamma} \frac{1}{\zeta-z} \dd \zeta \neq 0\right\}
\end{align*}
where $f_\gamma$ is the symbol function associated to $\hhn$, so $f_\gamma: z\mapsto e^\gamma z + e^{-\gamma}z^{-1}$ for $z\in\C$. It is clear that $\mathbb{W}$ is an ellipse with vertices at $\pm (e^\gamma + e^{-\gamma}) = \pm 2\cosh(\gamma)$. The first thing to realize is that \eqref{eq: L0 Lloyd} is symmetric with respect to $x=0$, that is $L_0(-x)=L_0(x)$ for any real $x$. On the other hand the original formula for $w(z)$ directly shows that
\begin{align*}
    \sgn\frac{\dd L_0}{\dd x}(x) = \sgn x.
\end{align*}
A Taylor expansion of \eqref{eq: L0 Lloyd} in $s$, considering the necessary branch cuts yields
\begin{align}
    L_0(x)=\begin{dcases}
        \arcosh(\vert x/2\vert) + \frac{x}{2(x^2-4)^{3/2}}s^2+\mathcal{O}(s^3), &\vert x \vert > 2\\
        0 + \frac{1}{\sqrt{4-x^2}}s+\mathcal{O}(s^4), &\vert x \vert < 2.
    \end{dcases}
    \label{eq: Lyap lloyd in out}
\end{align}

The last two observations prove the following result.
\begin{theorem}\label{thm: main Lloyd}
    Let $\hhn$ be a Hatano-Nelson Hamiltonian as in \eqref{eq: HN Hamiltonian} with \emph{i.i.d.} Cauchy distributed potentials of scale $s\in\R_{>0}$. Then the Lyapunov exponent associated to a frequency $\lambda \in \R$ satisfies
    \begin{align}
        \lambda \in \mathbb{W}^{o} \Longrightarrow L_\gamma(\lambda) < 0\nonumber \\
        \lambda \not\in \overline{\mathbb{W}} \Longrightarrow L_\gamma(\lambda) > 0
        \label{eq: thm lloyd}
    \end{align}
    up to an error that behaves linearly in $s$ for $\lambda \in \mathbb{W}$ and quadratically for $\lambda \not\in \mathbb{W}$ as $s\to 0$. More precisely this means that
    \begin{align*}
        L_\gamma(\lambda) = \tilde{L}_\gamma(\lambda) + C(s)
    \end{align*}
    where \eqref{eq: thm lloyd} holds for $\tilde{L}_\gamma(\lambda)$ and $\vert C(s)\vert =\mathcal{O}(s)$ for $\lambda \in \mathbb{W}^{o}$ and $\vert C(s)\vert =\mathcal{O}(s^2)$ for $\lambda \not\in \overline{\mathbb{W}}$ as $s\to 0$..
\end{theorem}
We remark at this point that \cref{thm: main Lloyd} characterizes exactly the transition that we first introduced in \cref{fig:intro} up to the behavior of the boundary of $\mathbb{W}$ (which is anyway of zero-measure). Furthermore, \eqref{eq: Lyap lloyd in out} provides an error estimation based on the scale of the parameter of the Cauchy distribution. This can be interpreted as a ``thickening'' of the boundary of $\mathbb{W}$ where no statement is made.
\subsection{Numerical simulations}
We consider now a couple of experiments regarding \cref{thm: main Lloyd}. In \cref{fig: Layp_many} we compute the Lyapunov exponent $L_\gamma(\lambda)$ for many realizations for $4$ different scale $s=0.1,0.5,1.0,2.0$ (in \cref{fig: cdfs} we show the cumulative distribution function for these parameters).
\cref{fig: Layp_many} clearly shows that the phase transition point $L_\gamma(\lambda)<0$ to $L_\gamma(\lambda)>0$ happens at the transition out of $\mathbb{W}$. The less the noise, the more accurate is this identification.
\begin{figure}
    \begin{subfigure}[t]{0.49\textwidth}
    \centering
\includegraphics[width=1\textwidth]{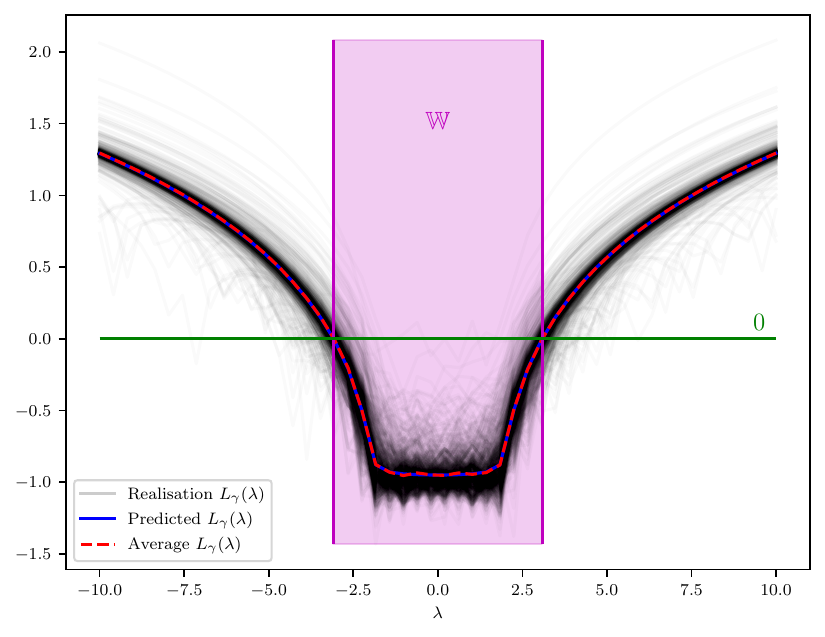}
    \caption{$s=0.1$}
    \label{fig: Lyap_many 0.1}
    \end{subfigure}\hfill
    \begin{subfigure}[t]{0.49\textwidth}
    \centering
\includegraphics[width=1\textwidth]{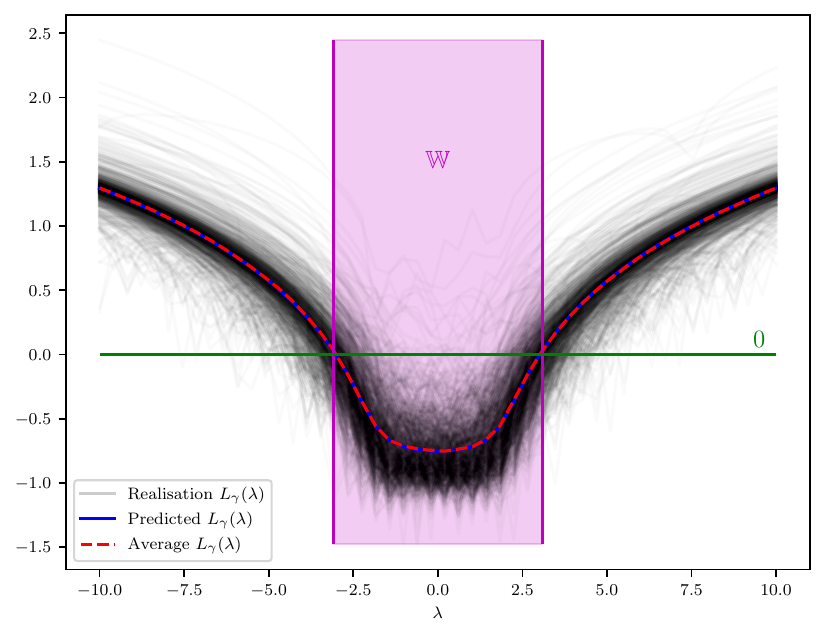}
    \caption{$s=0.5$}
    \label{fig: Lyap_many 0.5}
    \end{subfigure}\\
    \begin{subfigure}[t]{0.49\textwidth}
    \centering
\includegraphics[width=1\textwidth]{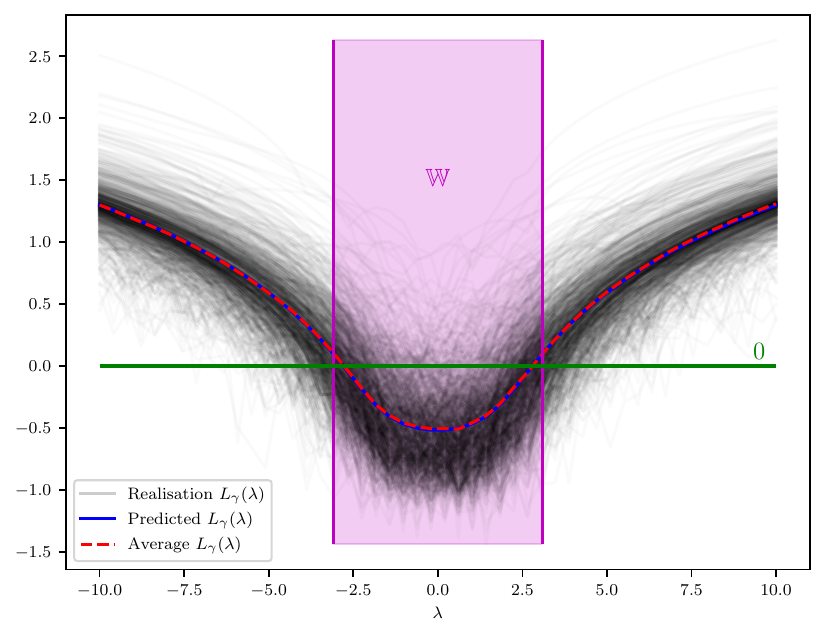}
    \caption{$s=1$}
    \label{fig: Lyap_many 1}
    \end{subfigure}\hfill
    \begin{subfigure}[t]{0.49\textwidth}
    \centering
\includegraphics[width=1\textwidth]{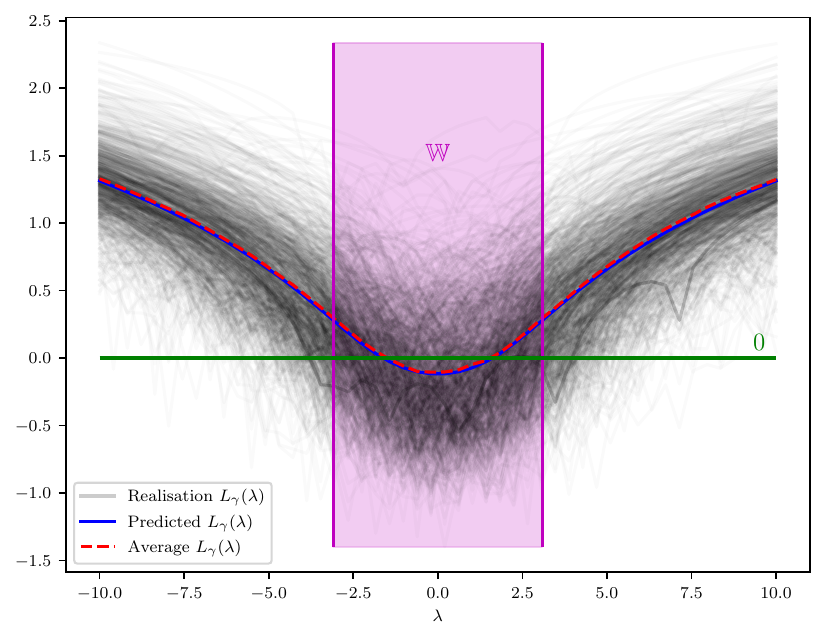}
    \caption{$s=2$}
    \label{fig: Lyap_many 2}
    \end{subfigure}\hfill
    \caption{Numerical simulation of $1000$ realizations of $\hhn^{\gamma=1}$ with potentials with Cauchy distributions of variable scale. On one hand one can appreciate the excellent agreement between the average computed Lyapunov exponent and the predicted value of \eqref{eq: layp herm to not herm} and \eqref{eq: L0 Lloyd}. On the other hand, the prediction power of \cref{thm: main Lloyd} remains accurate up to $s=1$. We remark that $\mathbb{P}[\vert X_{s=1}\vert > 1] = 0.5$ and $\mathbb{P}[\vert X_{s=2}\vert > 1] = 0.7$ for $X_{s}\sim$ Cauchy of scale $s$.}
    \label{fig: Layp_many}
\end{figure}
In \cref{fig: softargmax Lloyd} we look at $10^5$ realizations of a $N=100$ Hamiltonian $\hhn^1$  and consider the localization position of the eigenvectors depending on the position of the eigenvalues. To do this we compute the $\operatorname{softargmax}_N(x)$ defined as 
\begin{align*}
    \operatorname{softargmax}_\beta(x)=\sum_{i=0}^{N-1}
i \;
\frac{\exp\!\left(\beta (x_i - x_{\max})\right)}
{\sum_{j=0}^{N-1}\exp\!\left(\beta (x_j - x_{\max})\right)},
\qquad
x_{\max} = \max_j x_j
\end{align*}
of the eigenvectors. The phase transition is again apparent and in line with the findings of \cref{thm: main Lloyd}: while the eigenvalues lay within $\mathbb{W}$ the eigenvectors are localized at the edge of the system; once the eigenvalues lie outside of $\mathbb{W}$ Anderson localization takes over and there is no preferred localization position. 
\begin{figure}
    \includegraphics[width=0.5\textwidth]{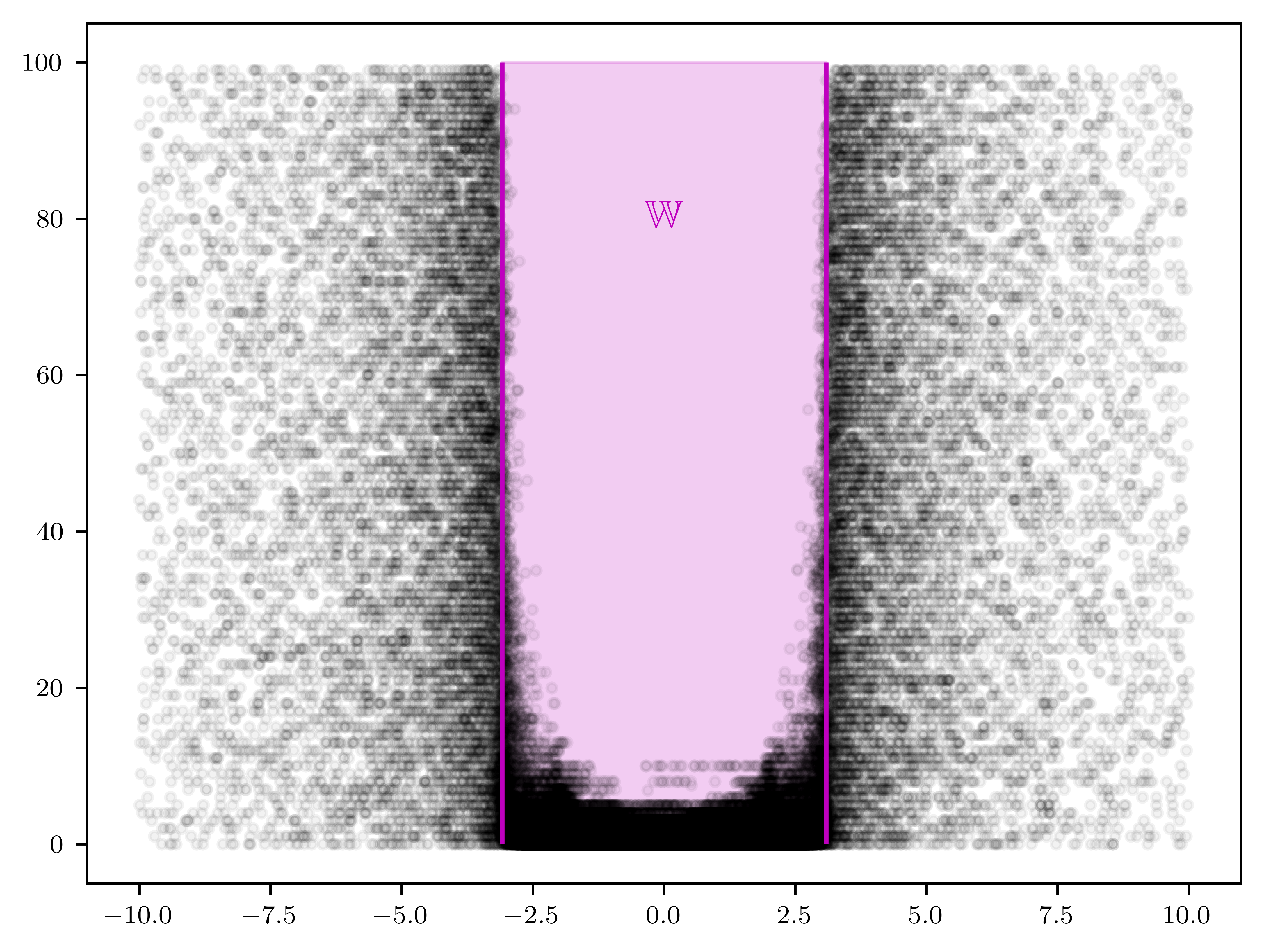}
    \caption{Grey dots represent pairs $(\lambda, \operatorname{softargmax}(v))$ where $(\lambda, v)$ is an eigenpair of one of 1000 samples of $\hhn^1$ for a Cauchy distribution of scale $s=0.1$. $\mathbb{W}$ represent the topological region associated to the deterministic $\hhn^1$. The transition from condensation to localization in the bulk is observed.}
    \label{fig: softargmax Lloyd}
\end{figure}
\section{General asymptotic results}\label{sec:gen asymptotic}
In this section we are going to lose the distribution assumption on the random potentials. As we have seen from \eqref{eq: layp herm to not herm} much can be said once the Hermitian case is understood, so we first summarize some known results.
\subsection{Weakly disordered Schrödinger model}
The theory of weakly disordered models is best developed for the Hermitian Schrödinger model
\begin{align}\label{eq: Schrodinger ham}
    H = \sum_i \vert i+1\rangle\langle i\vert + \vert i\rangle\langle i+1\vert + V_i \vert i\rangle\langle i\vert
\end{align}
where $V_i$ are the on-site potentials. We follow \cite[1.C, 6.D]{pastur.figotin1992Spectra} in stating now two definitions which give the general framework in which our main result holds.

\begin{definition}[Homogeneous Markov process]\label{def:hom_markov_proc}
    A real-valued random process $q(x)$ for $x\in\R$ is called \emph{homogeneous Markov process} if there is, on the set $\Omega$ of functions $f:\R\to\R$, a family of measures $\mu_{x,q}$ such that
    \begin{enumerate}
        \item if $\mathcal{F}_x^y$ denotes the $\sigma$-algebra generated by the random variables $\{q(z) : x \leqslant z \leqslant y\}$, then $\mu_{x,q}$ is a probability measure on the measure space $(\Omega, \mathcal{F}_x^\infty)$, and $\mu_{x,q}\{q(x) = q\} = 1$;
        \item for any $y \geqslant x \geqslant 0$ and $Q \in \mathcal{B}$ the functions
        \begin{align*}
            P(x, q; y, Q) = \mu_{x,q}\{q(y) \in Q\}
        \end{align*} is $\mathcal{B}(\R)$-measurable with respect to $q$;
        \item if $P\{\cdot \mid \mathcal{F}_1\}$
        denotes the conditional probability with respect to a fixed $\sigma$-subalgebra $\mathcal{F}_1$, we have, for any $z \geqslant y \geqslant x$ and $Q \in \mathcal{B}$:
        \begin{align*}
            \mu_{x,q}(\{q(z) \in Q \mid \mathcal{F}_x^y\}) = \mu_{y,q(y)}(\{q(z) \in Q\});
        \end{align*}
        \item for any $z \geqslant 0, y \geqslant x$ and $Q \in \mathcal{B}$ we have 
        \begin{align*}
            \mu_{x,q}(\{q(y) \in Q\}) = \mu_{x+z,q}(\{q(y+z) \in Q\}).
        \end{align*}
    \end{enumerate}
\end{definition}

\begin{definition}[Brownian motion model]\label{def:bm_model}
    Let $X(x)$ for $x\in\R$ be a homogenous Markov process on a compact $n$-dimensional Riemannian manifold $K$ and $Q:K\to\R$ a smooth function for which there exists a $k_0$ such that for any $X\in K$ there exists a $k\leq k_0$ so that $d^kQ(X)\neq 0$. Assume furthermore that 
    \begin{align*}
        \min_{X\in K} Q(X)=0,\quad \max_{X\in K}Q(X)=1.
    \end{align*}
    Then $q(x) \coloneqq Q(X(x))$ is called \emph{Brownian motion model}.
    
\end{definition}

To extricate the rather complex entanglement of conditions of \cref{def:hom_markov_proc,def:bm_model} we remark that easy examples as taking $q(x)$ to be \emph{i.i.d.} uniform random variables rescaled to $[0,1]$ satisfies \cref{def:bm_model}. The same holds for a truncated normal distribution.

Given these construct we can state the key building block of our final result as presented in \cite[§14.B]{pastur.figotin1992Spectra}.

\begin{proposition}[\citeauthor{pastur.figotin1992Spectra}]\label{thm: past fig}
    Let $H$ be the Schrödinger Hamiltonian \eqref{eq: Schrodinger ham} where $V_i = Q(X(i))$ is a Brownian motion model as in \cref{def:bm_model}, let $g=\sup_{X\in K}\vert Q(X)\vert$, and assume that $\mathbb{E}[Q(X(0))]=0$. Then for any $\lambda > 0$,
    \begin{align}
        L(\lambda) = \frac{1}{8\lambda}\hat{b}(2\sqrt{\lambda}) + \mathcal{O}(g^3)\quad\text{as}\quad g\to0.\label{eq: lyap result past fig}
    \end{align}
    Here
    $$
    b(x) \coloneqq \mathbb{E}[Q(X(0))Q(X(x))]
    $$
    is the correlation function of the process $Q(X(x))$, and $\hat{b}(x)$ is its Fourier transform.
\end{proposition}

\subsection{Phase transition}
The following theorem is our main result of this section and shows the link of the phase transition to the topological invariant.
\begin{theorem}
    Let $\hhn$ be as in \eqref{eq: HN Hamiltonian} with $V_i$ \emph{i.i.d.} random variables with $\E(V_0)=0$ and assume that $\var(V_0)=a\gamma$ for any $a>0$ and $\operatorname{ess\ sup}(V_0)\to 0$ as $\sigma\to 0$, then
    \begin{align}
        L_\gamma(\lambda) > 0 \Leftrightarrow \lambda \not\in \mathbb{W}
    \end{align}
    in the limit $\gamma \to 0$. The error is quadratic in $\gamma$ for $a=16$.
\end{theorem}
\begin{proof}
    Let $\sigma^2=\var(V_0)$. Then, given the \emph{i.i.d.} assumption on $V_i$ it is straightforward to see $\E[V_0 V_i]=\sigma^2\delta_{i,0}$ and consequently $\hat{b}(i)=\sigma^2$. Using \eqref{eq: lyap result past fig} from \cref{thm: past fig} we conclude
    $$
    L_0(\lambda) = \frac{\sigma^2}{8\lambda} + \mathcal{O}(\operatorname{ess\ sup}(V_0)).
    $$
    Finally \eqref{eq: layp herm to not herm} together with a Taylor expansion of $x/\cosh(x)$ yields the result.
\end{proof}
\subsection{Numerical simulations}
We consider as easiest example a uniform random variable on the interval $[-\sigma\sqrt{3},\sigma\sqrt{3}]$ for some variance $\sigma^2>0$. In \cref{fig:convergence} we look at the average mean error that is the mean of 
\begin{align}
    \vert L_\gamma(2\cosh(\gamma)) \vert
    \label{eq: mean error}
\end{align}
over many realizations of $\hhn$ for a rather large system of $N=10^5$.

\begin{figure}
    \includegraphics[width=0.5\textwidth]{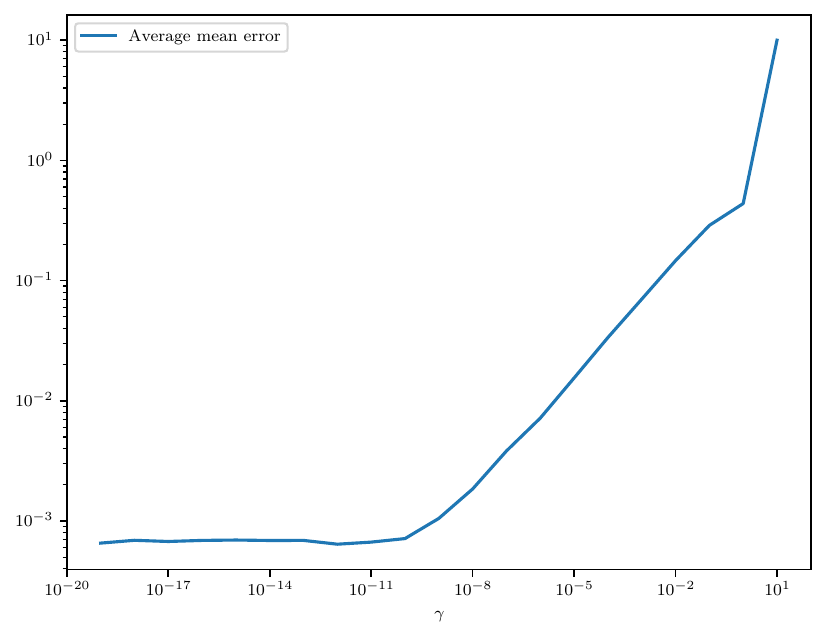}
    \caption{We show the convergence of average mean error as in \eqref{eq: mean error} for a system of $N=10^5$ over $10$ runs. The stagnation of the error is due to the finite size of the system.}
    \label{fig:convergence}
\end{figure}

We remark that the poor convergence after $10^{-10}$ is due to the finiteness of the system as such small Lyapunov exponent can only be captured on systems of size $1/L_\gamma(\lambda)$.

We also point out that this result appears to be much weaker than the one in \cref{sec:Lloyd model}, but it sheds light on the non-Hermitian to Hermitian transition which is fundamentally different in nature as we mentioned in the introduction and will further expand in the following conclusions.

\section{Conclusion}
In this paper we have provided a rigorous characterization of the transition between two fundamentally distinct localization mechanisms in disordered non-Hermitian systems: the non-Hermitian skin effect and Anderson localization. We have done so in two very different regimes. The Lloyd model is very precise and offers a unique way to prove this statement exactly. The weak-disorder limit shows us how the transition from non-Hermitian to Hermitian happens in a rather general framework. In particular this last regime handles the fundamental difference that the NHSE introduces to disordered systems: a minimal non-zero disorder strength is needed to generate Anderson localization.

Through the relation of the Lyapunov exponent $L_\gamma(\lambda) = L_0(\lambda) - \gamma$, the non-Hermitian problem reduces to the Hermitian one up to a constant shift, and the topological winding region $\mathbb{W}$ of the unperturbed system emerges as the \emph{a priori} spectral discriminant: eigenstates associated to frequencies inside $\mathbb{W}$ are skin-effect localized at the boundary, while those outside are Anderson localized in the bulk.

Several directions for future research present themselves naturally. First, an extension to two and higher-dimensional non-Hermitian systems, where the topology of $\mathbb{W}$ becomes richer and the interplay with Anderson localization remains largely uncharted. Second, a study of the boundary of $\mathbb{W}$: at these critical frequencies the system is expected to exhibit delocalized or critical states, and quantifying their fractal or multifractal structure seems worth researching. Third, the extension from \emph{i.i.d.} to correlated disorder originally considered in \cite{ammari.barandun.ea2024Stability}: numerical results show exactly the same behaviors.

\section{Acknowledgements}
The author thanks the Swiss National Science Foundation (SNSF) for its support through the grant \href{https://data.snf.ch/grants/grant/235080}{235080}.\\
The author thanks Prof. H. Ammari, Prof. S. Mayboroda, Prof. S. G. Johnson and A. Uhlmann for insightful discussions and Prof. C. Bordenave for precious advice on an earlier attempt to this problem and hosting his research visit.
\printbibliography

\end{document}